\documentclass[letterpaper, 10 pt, conference]{ieeeconf}  

\IEEEoverridecommandlockouts                              

\usepackage{bm}
\usepackage{nicefrac}
\usepackage{amssymb}
\usepackage{amsmath}

\usepackage{amsthm}
\usepackage{cite}
\usepackage{xcolor}
\usepackage{graphicx}
\usepackage{url}
\usepackage{hyperref}
\usepackage{svg}
\usepackage{booktabs}
\usepackage{siunitx}
\usepackage{multirow}
\usepackage{colortbl,xcolor} 
\usepackage{threeparttable}  
\usepackage[caption=false,font=normalsize,labelfont=sf,textfont=sf]{subfig}

\usepackage{balance}
\newtheoremstyle{exampstyle}
  {3pt} 
  {3pt} 
  {\itshape} 
  {} 
  {\bfseries} 
  {.} 
  {.5em} 
  {} 

\theoremstyle{exampstyle} 
\newtheorem{definition}{Definition}
\newtheorem{lemma}{Lemma}
\newtheorem{theorem}{Theorem}

\newtheorem{assumption}{Assumption}

\newtheorem{problem}{Problem}

\theoremstyle{plain}

\definecolor{mumred}{RGB}{222,33,77}
\definecolor{mumgreen}{RGB}{0, 140, 0}
\definecolor{mumblue}{RGB}{0, 100, 222}
\definecolor{mumpurple}{RGB}{128, 0, 128}

\title{\LARGE \bf
Risk-Aware Robot Control in Dynamic Environments Using \\
Belief Control Barrier Functions
}

\author{Shaohang Han, Matti Vahs and Jana Tumova
	\thanks{This work was partially supported by the Wallenberg AI, Autonomous
		Systems and Software Program (WASP) funded by the Knut and Alice
		Wallenberg Foundation. }
	\thanks{The authors are with the Division of Robotics, Perception and Learning, School of Electrical Engineering and Computer Science, KTH Royal Institute of Technology, Stockholm, Sweden and also affiliated with Digital Futures. Mail addresses: {\{\tt\small shaohang, vahs, tumova\}}
		{\tt\small @kth.se}}%
}

\begin{document}

\maketitle
\thispagestyle{empty}
\pagestyle{empty}

\begin{abstract}

Ensuring safety for autonomous robots operating in dynamic environments can be challenging due to factors such as unmodeled dynamics, noisy sensor measurements, and partial observability. To account for these limitations, it is common to maintain a belief distribution over the true state. This belief could be a non-parametric, sample-based representation to capture uncertainty more flexibly. In this paper, we propose a novel form of Belief Control Barrier Functions (BCBFs) specifically designed to ensure safety in dynamic environments under stochastic dynamics and a sample-based belief about the environment state. Our approach incorporates provable concentration bounds on tail risk measures into BCBFs, effectively addressing possible multimodal and skewed belief distributions represented by samples. Moreover, the proposed method demonstrates robustness against distributional shifts up to a predefined bound. We validate the effectiveness and real-time performance (approximately \SI{1}{\kilo\hertz}) of the proposed method through two simulated underwater robotic applications: object tracking and dynamic collision avoidance. 

\end{abstract}

\section{Introduction}

When deploying autonomous robots in real-world settings, it is crucial to ensure they meet safety specifications; this is where safety-critical control comes into play. A popular approach in this domain is the use of Control Barrier Functions (CBFs), which effectively synthesize safe control inputs via CBF-based quadratic programs (CBF-QPs) \cite{ames2019control}. Due to their computational efficiency, CBF-QPs are widely employed to ensure safety in robotics \cite{wang2022safety, du2023reinforcement,xu2025learning, zhao2023stable,9981593}. Many works have extended the CBF-based framework to compensate for the uncertain dynamics in the real world \cite{clark2021control,black2023safety,singletary2022safe,lederer2024safe,cohen2022robust}. Specifically, stochastic CBFs (SCBFs) ensure system safety with probability one under stochastic dynamics represented by stochastic differential equations (SDEs) \cite{clark2021control}. Additionally, the authors in \cite{black2023safety} propose risk-aware CBFs that bound the risk of the stochastic system becoming unsafe. 

However, in extreme environments such as underwater settings, noisy sensor measurements, partial observability, and uncertain dynamics can also make it challenging to determine the true system state. To account for this, robotic software stacks commonly employ a perception module that provides the system’s \emph{belief}, which is a probability distribution over possible states \cite{thrun2006probabilistic}. As an example, consider an autonomous underwater vehicle (AUV) operating in a dynamic environment in the presence of a moving object, such as another AUV. Depending on the task, the AUV may be required to keep this object within its field of view (FoV) for continuous tracking or ensure collision avoidance to maintain safe operation. Unfortunately, the measurements could be noisy due to occlusions, reflections, or temporary low visibility, as exemplified in Fig.~\ref{fig:FirstPage}.
In such situations, the belief is commonly represented as a set of finite state samples. This non-parametric form can capture distributions that could be multimodal and skewed \cite{chun20243d,dorner2024smooth,masmitja2023dynamic}.

To incorporate an uncertain state within the CBF-based framework, existing work has assumed Gaussian beliefs \cite{vahs2023belief,wei2024confidence,li2023moving}, or bounded state estimation errors \cite{dean2021guaranteeing,agrawal2022safe,zhang2022control}. Yet, it is also important to consider beliefs that are unbounded and represented by finite samples. This introduces two key challenges: (i) quantifying uncertainty from finite samples, and (ii) integrating this uncertainty into CBF safety constraints. In parametric belief models (e.g., Gaussian distributions), uncertainty quantification is straightforward due to closed-form expressions for tail risk measures such as Value-at-Risk ($\mathrm{VaR}$) \cite[Example 2.14]{mcneil2015quantitative} and Conditional-Value-at-Risk ($\mathrm{CVaR}$) \cite[Example 2.18]{mcneil2015quantitative}. These risk measures are particularly effective for reasoning about tail events that may make the robot unsafe \cite{akella2024risk}. In contrast, sample-based beliefs preclude direct analytical computation, requiring alternative approaches to address these challenges.

\begin{figure}[t]
    \centering
    \includegraphics[width=0.8\columnwidth]{
    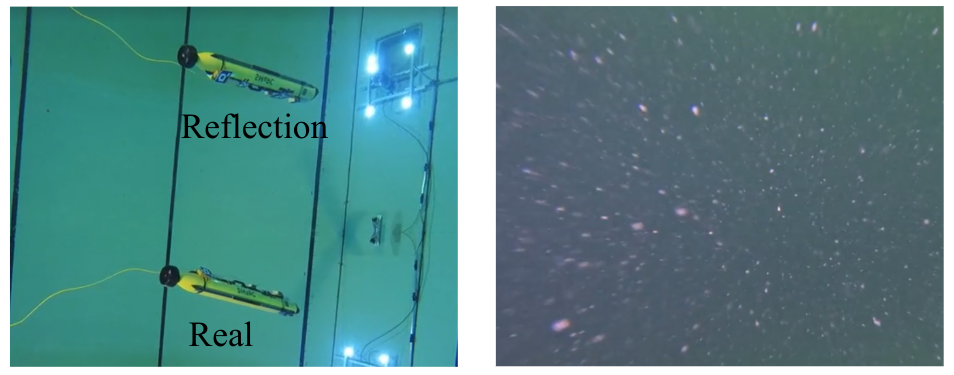} 
    \caption{\textbf{Left:} Reflections near the water surface can lead to ambiguous detections.
    \textbf{Right:} Underwater bubbles cause temporary low visibility.
    }
    \label{fig:FirstPage}
    \vspace{-0.5cm}
\end{figure}

In the literature, only a few works address sample-based beliefs in CBF-based frameworks. Recent works \cite{chriat2024wasserstein, long2024sensor} construct Wasserstein ambiguity sets from state samples and leverage distributionally robust optimization (DRO) to reformulate CBF constraints with probabilistic guarantees. By adopting $\mathrm{CVaR}$ as a risk measure, the authors derive tractable convex reformulations that preserve the convex optimization structure of the control synthesis problem. 
While this approach can be effective, it may become computationally expensive as the number of samples increases.
On the other hand, \cite{do2024probabilistically} proposes a scenario-based method to ensure probabilistic guarantees by satisfying a finite number of CBF constraints corresponding to the samples. However, only linear dynamics are considered in this approach. 
Our previous work \cite{vahs2024risk} proposed belief CBFs (BCBFs), constructed using the sample-based lower bound of $\mathrm{CVaR}$, to address localization uncertainty of the robot. 
In this work, we generalize the use of sample-based lower bounds of risk measures within BCBFs to handle beliefs about the environment, including moving objects. 

\textbf{Contribution:} In this paper, we propose a generalized formulation of BCBF for sample-based beliefs and stochastic dynamics, leveraging concentration bounds on two types of tail risk measures—$\mathrm{VaR}$, $\mathrm{CVaR}$—as established in \cite{vincent2024guarantees}. While these bounds have been shown to be effective for performance evaluation \cite{vincent2024guarantees}, we extend their application to safety-critical control synthesis within the CBF framework. Furthermore, the proposed BCBFs exhibit robustness to bounded distribution shifts, which may be resulted from model mismatches in Bayes filtering within the perception module. Finally, we demonstrate the effectiveness and computational efficiency of our method in two AUV applications: object tracking and dynamic collision avoidance.

\section{Preliminaries}
We consider stochastic dynamics for both the robot and the moving object, represented as stochastic differential equations (SDEs). The robot dynamical model is given by
\begin{align}
\label{eq:sde_ro}
    \mathrm{d}\bm{x} = \left(\bm{f}(\bm{x})+\bm{g}(\bm{x})\bm{u}\right)\mathrm{d}t + \bm{\sigma}(\bm{x})\mathrm{d}\bm{z}
\end{align}
where $\bm{x} \in \mathcal{X} \subseteq \mathbb{R}^{n_x}$ is the robot state, $\bm{u} \in \mathcal{U} \subseteq \mathbb{R}^m$ denotes the control input, the functions $\bm{f}:\mathbb{R}^{n_x}\mapsto\mathbb{R}^{n_x}$ and $\bm{g}:\mathbb{R}^{n_x}\mapsto\mathbb{R}^{n_x \times m}$, and $\bm{z} \in \mathbb{R}^{n_z}$ is a $n_z$-dimensional Brownian motion. The motion of the object is governed by
\begin{align}
\label{eq:sde_obj}
    \mathrm{d}\bm{o} = \bm{\xi}(\bm{o})\mathrm{d}t + \bm{d}(\bm{o})\mathrm{d}\bm{w},
\end{align}
where $\bm{o} \in \mathcal{S} \subseteq \mathbb{R}^{n_o}$ is the object state, $\bm{\xi}:\mathbb{R}^{n_o}\mapsto\mathbb{R}^{n_o}$ is the transition function. $\bm{w} \in \mathbb{R}^{n_w}$ is a $n_w$-dimensional Brownian motion. We assume the diffusion terms $\bm{\sigma}(\bm{x})$ and $\bm{d}(\bm{o})$ are globally Lipschitz non-degenerate diagonal matrices, and that the drift terms $\left(\bm{f}(\bm{x}) + \bm{g}(\bm{x})\bm{u}\right)$ and $\bm{\xi}(\bm{o})$ are locally Lipschitz functions that can have jumps. Under these assumptions,~(\ref{eq:sde_ro}) and~(\ref{eq:sde_obj}) admit unique global strong solutions \cite{leobacher2017strong}.

\subsection{Stochastic Control Barrier Functions}

Consider the SDEs defined in~(\ref{eq:sde_ro}) and~(\ref{eq:sde_obj}), a safe set ${\mathcal{C} \subseteq \mathbb{R}^{n_x} \times \mathbb{R}^{n_o}}$ is the closed zero super-level set of a twice continuously differentiable function $h:\mathbb{R}^{n_x} \times \mathbb{R}^{n_o} \mapsto \mathbb{R}$, which is defined as
\begin{equation}
    \begin{aligned}
        \mathcal{C} &:= \{ (\bm{x},\bm{o}) \in \mathcal{X} \times \mathcal{O} \mid h (\bm{x},\bm{o}) \geq 0 \}, \\
        \partial \mathcal{C} &:= \{ (\bm{x},\bm{o}) \in \mathcal{X} \times \mathcal{O} \mid h (\bm{x},\bm{o}) = 0 \}.
    \end{aligned}
    \label{eq:safe_set}
\end{equation}

\begin{definition}
    A safe set $\mathcal{C}$ is forward invariant with respect to the systems~(\ref{eq:sde_ro}) and~(\ref{eq:sde_obj}) if for every initial condition ${(\bm{x}_{t_0},\bm{o}_{t_0}) \in \mathcal{C}}$ it holds that $(\bm{x}_t,\bm{o}_{t}) \in \mathcal{C}, \forall t \geq t_0$ with probability 1.
\end{definition}

Stochastic CBFs, including reciprocal CBFs (RCBFs) and zeroing CBF (ZCBFs) forms, were proposed in \cite{clark2021control} ensuring forward invariance of the safe set $\mathcal{C}$. Then in \cite{so2023almost},  the authors provided refined formulations specifically for stochastic ZCBFs. Here, we express these formulations in a form that explicitly accounts for both the robot and the object, as formalized below.

\begin{definition}
\label{def:nszcbf}
    Given a safe set $\mathcal{C}$ defined by ~(\ref{eq:safe_set}), the function $h (\bm{x},\bm{o})$ serves as a zeroing stochastic CBF (ZSCBF) for system~(\ref{eq:sde_ro}) and~(\ref{eq:sde_obj}), if $\forall (\bm{x},\bm{o}) \in \mathcal{C}$, there exist an extended class-$\kappa$ function $\alpha$ and a control input $\bm{u}\in\mathcal{U}$ such that
    \begin{equation}
    \label{ineq:nszcbf}
        \begin{aligned}
        &\frac{\partial h}{\partial \bm{x}} \left(f(\bm{x}) + g(\bm{x})\bm{u}\right) 
        + \frac{1}{2} \operatorname{tr} 
        \left( \bm{\sigma}^\top \frac{\partial^2 h}{\partial \bm{x}^2} \bm{\sigma} \right)
        - \frac{\lVert \frac{\partial h}{\partial \bm{x}} \bm{\sigma}(\bm{x}) \rVert_2^2}{h} \\
        &+ \frac{\partial h}{\partial \bm{o}} \bm{\xi}(\bm{o})
        + \frac{1}{2} \operatorname{tr} 
        \left( \bm{d}^\top \frac{\partial^2 h}{\partial \bm{o}^2} \bm{d} \right)
        - \frac{\lVert \frac{\partial h}{\partial \bm{o}} \bm{d}(\bm{o}) \rVert_2^2}{h} \\
        &\geq - h(\bm{x},\bm{o})^2 \alpha\left(h(\bm{x},\bm{o})\right).
        \end{aligned}
    \end{equation}
\end{definition}

\begin{theorem}[Corollary 11, \cite{so2023almost}]
\label{th:nszcbf}
    Suppose there exist a function $h$ and a locally Lipschitz control input $\bm{u}$ satisfying Def.~\ref{def:nszcbf}, if $(\bm{x}_{t_0},\bm{o}_{t_0}) \in \mathcal{C}$, then ${\Pr \left[ (\bm{x}_t,\bm{o}_t) \in \mathcal{C}, \forall t \geq t_0 \right] = 1}$.
\end{theorem}

We present the following lemma, which will later be used in the derivation of our main theoretical result.

\begin{lemma}
\label{lm:non_smooth}
    Suppose $h(\bm{x},\bm{o})$ is continuous but non-smooth, the result in Theorem~\ref{th:nszcbf} still holds under the stochastic dynamics in~(\ref{eq:sde_ro}) and~(\ref{eq:sde_obj}). \vspace{-2.5mm}
\end{lemma}
\begin{proof}
    The proof is provided in Appendix~\ref{ap:proof_nszcbf}.
\end{proof}


\subsection{Tail Risk Measures}
Consider a scalar random variable $Y \in \mathcal{Y} \subseteq \mathbb{R}$ with probability density function ${\mathrm{PDF}(y):=p(Y=y)}$ and cumulative density function ${\mathrm{CDF}(y):=\Pr[Y\leq y]}$, a risk measure is a function that maps $Y$ to a real value: $\rho(Y):\mathcal{Y}\mapsto\mathbb{R}$. We focus on two tail risk measures: Value-at-Risk ($\mathrm{VaR}$) and Conditional-Value-at-Risk ($\mathrm{CVaR}$), both defined at the lower quantile of the distribution, similar to \cite{singletary2022safe}.
\begin{definition}
The Value-at-Risk ($\mathrm{VaR}$) of a random variable 
$Y \in \mathbb{R}$ is the lower quantile at level $\tau$ where $\tau \in (0,1)$, formally defined as
$\mathrm{VaR}_{\tau}(Y) = \sup_{y} \left\{ y \mid \mathrm{CDF}(y) \leq \tau \right\}.$

Serving as a loose lower bound of $\mathrm{VaR}$, the Conditional-Value-at-Risk ($\mathrm{CVaR}$) of a random variable 
$Y \in \mathbb{R}$ at level $\tau\in(0,1]$ is
$\mathrm{CVaR}_{\tau}(Y) := \frac{1}{\tau} \int_{0}^{\tau} \mathrm{VaR}_{\nu}(Y) \, d\nu.$
\end{definition}

Both $\mathrm{VaR}$ and $\mathrm{CVaR}$ are closely related to chance constraints. Specifically, they satisfy the following inequality:
\begin{equation}
\label{eq:chance_constr}
    0 \leq \mathrm{CVaR}_{\tau}(Y) \leq \mathrm{VaR}_{\tau}(Y) \Rightarrow \Pr[Y\geq0] \geq 1-\tau,
\end{equation}
where $\tau$ is a user-defined risk level, typically set to small values such as $0.1$. Unlike the expected value $\mathbb{E}[Y]$, which computes the mean without accounting for the distribution's shape or tail behavior, tail risk measures $\mathrm{VaR}_{\tau}(Y)$ and $\mathrm{CVaR}_{\tau}(Y)$ are more effective in capturing tail events, making them particularly suitable for heavy-tailed, skewed, or multimodal distributions \cite{mcneil2015quantitative}.

\subsection{Concentration Bounds on Risk Measures}
The following lemma is adapted from \cite{vincent2024guarantees} by consolidating Theorems 1, 2, and 3 from \cite{vincent2024guarantees} for brevity. We have modified the original upper bounds from \cite{vincent2024guarantees} to serve as lower bounds in this work.

\begin{lemma}
\label{lm:lbs}

Let $Y^{(1)},\dots, Y^{(N)}$ be i.i.d. samples of a random variable $Y$, with ${\Pr[Y \geq Y_{\mathrm{lb}}]=1}$ for some finite $Y_{\mathrm{lb}}$. Let $\eta^{1:N}$ be the order statistics, i.e. $Y^{(1)},\dots, Y^{(N)}$ in descending order. $\eta^{k}$ denotes the $k$th smallest sample. 

\textbf{VaR Bound}: Suppose ${N\geq\lceil\ln\delta/\ln(1-\tau)\rceil}$ for ${\tau,\delta\in(0,1)}$. Let $k$ be the smallest index satisfying $\mathrm{Bin}\left(k-1;N,(1-\tau) \right)\geq1-\delta$, where $\mathrm{Bin}(k;m,p)$ denotes binomial $\mathrm{CDF}$ with $m$ trails, success probability $p$, $k$ successes. Then the lower bound on $\mathrm{VaR}$ is
\begin{equation}
\label{eq:lower_var}
    \underline{\mathrm{VaR}_\tau} := \eta^{k}.
\end{equation}

\textbf{CVaR Bound}: Suppose $N\geq\lceil-\frac{1}{2}\ln\delta/\tau^2\rceil$, for $\tau\in(0,1]$, $\delta \in (0, 0.5]$. Let $k$ be the smallest index satisfying ${\frac{k}{N}-\epsilon-1+\tau\geq0}$ with $\epsilon=\sqrt{-\mathrm{ln}\delta/2N} $. Then the lower bound on $\mathrm{CVaR}$ is
    \begin{equation}
        \begin{aligned}
            \underline{\mathrm{CVaR}_{\tau}} := \frac{1}{\tau} \bigg[ \epsilon b 
            &+ \left( \frac{k}{N} - \epsilon -1 + \tau \right) \eta^{k} \\
            &+ \frac{1}{N} \sum_{i=k+1}^{N} \eta^{i} \bigg].
        \end{aligned}
    \end{equation}
We can recover the lower bound on the expected value ${\underline{\mathbb{E}}=\underline{\mathrm{CVaR}_1}}$, when $\delta \in (0, 0.5]$ and $N\geq\lceil-\frac{1}{2}\ln\delta\rceil$. \vspace{0.5mm}

\textbf{Probabilistic Guarantee}: For $\rho_\tau$ being $\mathrm{VaR}_\tau$ and $\mathrm{CVaR}_\tau$, we have $\Pr[\underline{\rho_{\tau}}(Y^{(1)},\dots,Y^{(N)})\leq\rho_{\tau}(Y)]\geq1-\delta$.
\end{lemma}

Lemma~\ref{lm:lbs} gives the sample-based lower bounds on the risk measures, including $\mathrm{VaR}$, $\mathrm{CVaR}$ and $\mathbb{E}$, as illustrated in Fig.~\ref{fig:fig2a}.  In~\cite{vincent2024guarantees}, these bounds are used to verify the safety of a given controller. By contrast, we leverage these bounds to synthesize a safety controller. Specifically, we focus on the tail risk measures $\mathrm{VaR}$ and $\mathrm{CVaR}$, as they account for rare but critical events. We benchmark the performance of $\mathrm{VaR}$, $\mathrm{CVaR}$, and $\mathbb{E}$ by experiments presented later.

\begin{figure}[t]\captionsetup[subfigure]{font=scriptsize}
\centering
\subfloat[]{\includegraphics[width=0.48\columnwidth]{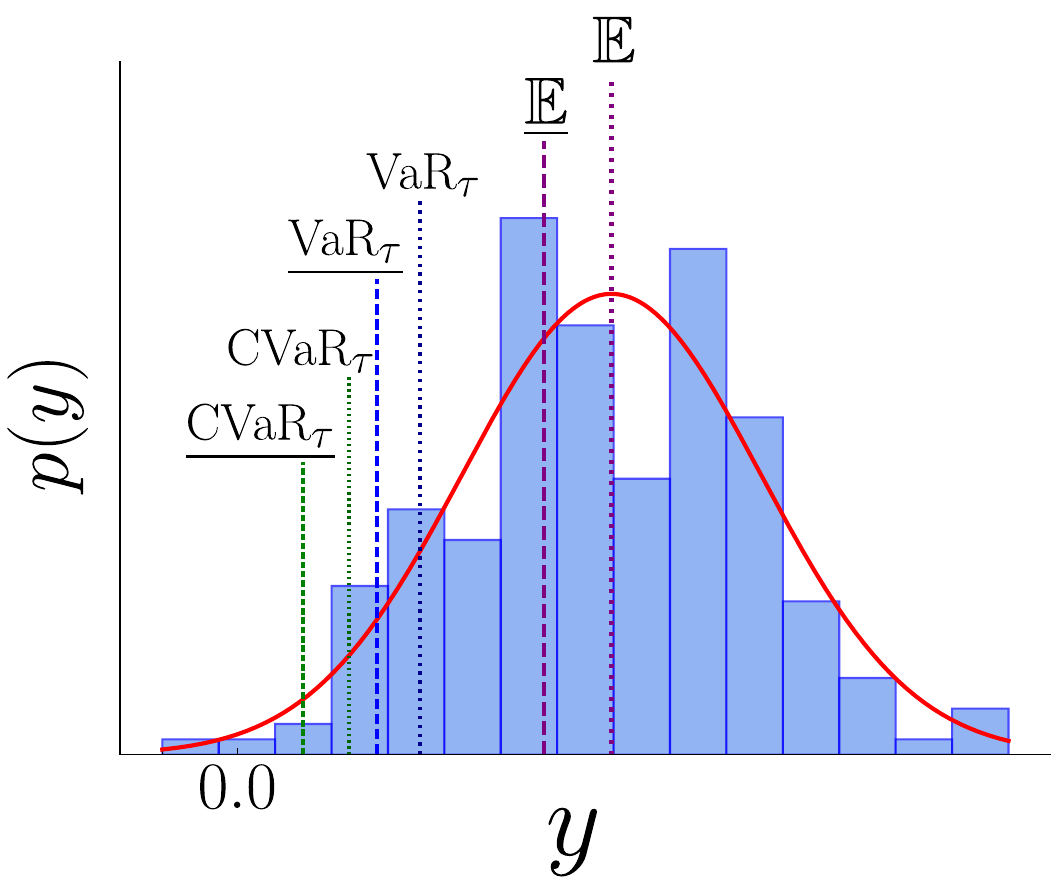} \label{fig:fig2a}}
\subfloat[]{\includegraphics[width=0.48\columnwidth]{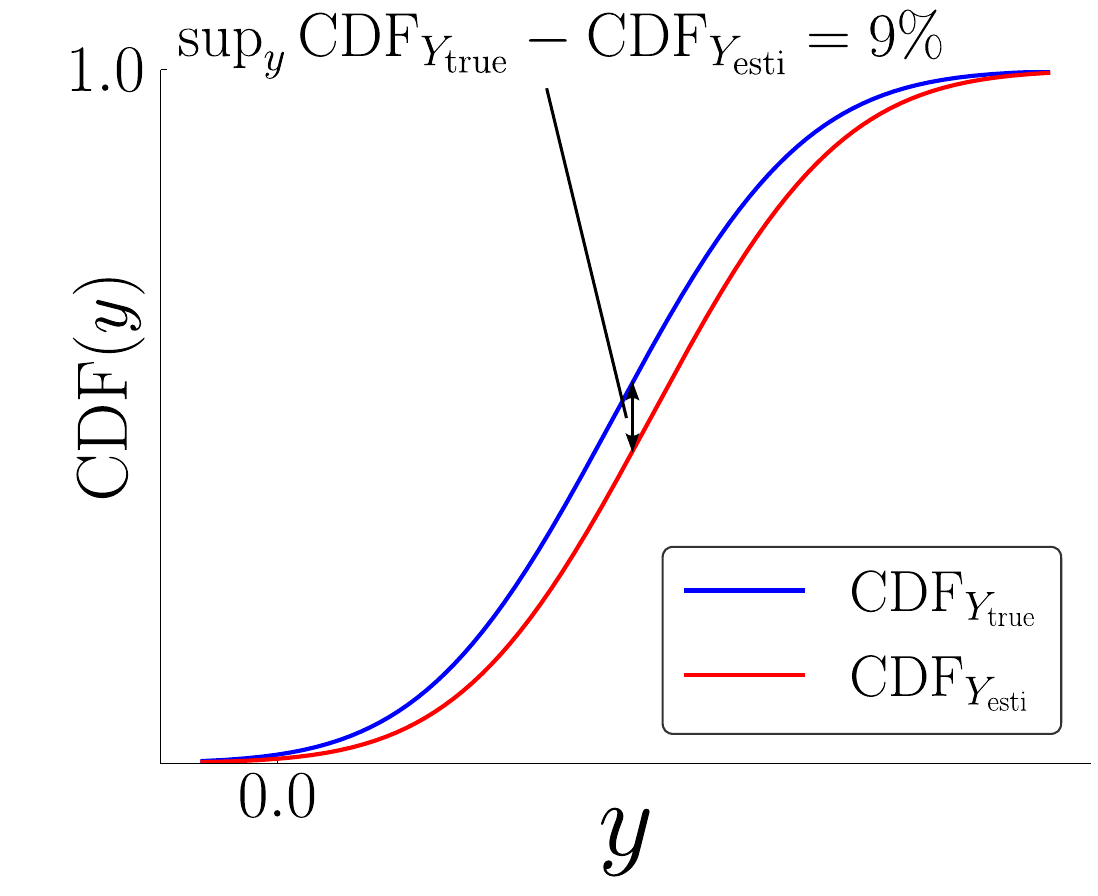} \label{fig:fig2b}}
\caption{
(a): The $\mathrm{PDF}$ of Gaussian distribution $\mathcal{N}(0.5, 0.2^2)$. We draw $N=1000$ samples, illustrated by the histogram with 15 bins. Both the real values of the risk measures and their lower bounds are shown. The lower bounds are computed using $\tau=0.1$ and $\delta=0.05$. (b): The estimated distribution is $\mathcal{N}(0.5, 0.2^2)$, while the true distribution is $\mathcal{N}(0.45, 0.2^2)$, showing a distribution shift.
}
\vspace{-0.5cm}
\end{figure}

\subsection{Robust Bounds on Risk Measures}
We can draw samples from the estimated distribution $Y_{\mathrm{esti}}$ but it may not fully match the true distribution $Y_{\mathrm{true}}$. This requires us to derive bounds that are robust to distribution shifts. The following lemma provides such bounds, combining Corollaries 4, 5, and 6 from \cite{vincent2024guarantees}.
\begin{lemma}
\label{lm:dist_robust}
    Let $Y_{\mathrm{esti}}^{(1)},\dots, Y_{\mathrm{esti}}^{(N)}$ be i.i.d. samples of a random variable $Y_{\mathrm{esti}}$, with ${\Pr[Y_{\mathrm{esti}}\geq Y_{\mathrm{lb}}]=1}$ for some finite $Y_{\mathrm{lb}}$. Let $\eta_{\mathrm{esti}}^{1:N}$ be the order statistics, i.e. $Y_{\mathrm{esti}}^{(1)},\dots, Y_{\mathrm{esti}}^{(N)}$ in descending order. $\eta_{\mathrm{esti}}^{k}$ denotes the $k$th smallest sample. Assuming $\sup_y \mathrm{CDF}_{Y_{\mathrm{true}}}(y) - \mathrm{CDF}_{Y_{\mathrm{esti}}}(y)  \leq \ell$, with ${\ell\in[0,\tau)}$.
    
    \textbf{Robust VaR Bound}: The $\ell$-robust $\mathrm{VaR}$ bound is
    \begin{align}
    \underline{\mathrm{VaR}_{\tau}^{\ell}}(Y_{\mathrm{esti}}^{(1)},\dots,Y_{\mathrm{esti}}^{(N)})  = \underline{\mathrm{VaR}_{\tau-\ell}}.
    \end{align}
    
    \textbf{Robust CVaR Bound}: Replacing $\epsilon$ by $\epsilon'=\epsilon-\ell$, we have the $\ell$-robust $\mathrm{CVaR}$ bound:
    \begin{equation}
        \begin{aligned}
            \underline{\mathrm{CVaR}_{\tau}^{\ell}} := \frac{1}{\tau} \bigg[ \epsilon' b 
            &+ \left( \frac{k}{N} - \epsilon' - 1 + \tau \right) \eta^k_{\mathrm{esti}} \\
            &+ \frac{1}{N} \sum_{i=k+1}^{N} \eta^{i}_{\mathrm{esti}} \bigg].
        \end{aligned}
    \end{equation}
We can obtain a similar robust bound $\underline{\mathbb{E}^\ell}=\underline{\mathrm{CVaR}_1^{\ell}}$. \vspace{0.5mm}

\textbf{Probabilistic Guarantee}: For $\rho_\tau^\ell$ being $\mathrm{VaR}_\tau^\ell$ and $\mathrm{CVaR}_\tau^\ell$, we have $\Pr[\underline{\rho_{\tau}^\ell}(Y_{\mathrm{esti}}^{(1)},\dots,Y_{\mathrm{esti}}^{(N)})\leq\rho_{\tau}(Y_{\mathrm{true}})]\geq1-\delta$.
\end{lemma}

Lemma~\ref{lm:dist_robust} provides robust bounds to mitigate distribution shifts when the $\mathrm{CDF}_{Y_{\mathrm{true}}}$ is above $\mathrm{CDF}_{Y_{\mathrm{esti}}}$,  as illustrated in Fig.~\ref{fig:fig2b}. 
This may result in a higher probability of $Y_{\mathrm{true}} <0$, which is undesirable in our setting.
Distribution shifts may arise due to the sim-to-real gap, as in \cite{vincent2024guarantees}. In this paper, we use these robust bounds to address distribution shifts in Bayesian filters, which may result from mismatches in the transition or measurement models.

\vspace{-0.5mm}
\section{Problem Formulation}
We consider a robot modeled by~(\ref{eq:sde_ro}) and a moving object modeled by~(\ref{eq:sde_obj}), both operating within a shared workspace. The robot should satisfy task specifications involving the object, such as tracking or collision avoidance, as discussed in the application examples later in the paper. We use $O$ to denote the random vector representing the object belief.

\begin{assumption}
\label{as:sample_belief}
    The controller has access to a finite set of i.i.d samples from the object's belief $O$, denoted as $\bm{o}^{(1)},\dots,\bm{o}^{(N)}$, where \(N\) represents the number of samples.
\end{assumption}

In practice, these samples can be drawn from belief distributions generated by Bayesian filtering techniques, such as a particle filter~\cite{chun20243d,dorner2024smooth,masmitja2023dynamic} or an extended Kalman filter~\cite{thrun2006probabilistic}.

As a function of the object belief $O$, the belief about $h$ is represented by $H(\bm{x}, O)$. We also have samples $ h^{(i)}(\bm{x}, \bm{o}^{(i)}) $ for $ i = 1, \dots, N $. In this paper, we aim to satisfy the chance constraint ${\Pr[H \geq 0] \geq 1-\tau}$, which ensures the safety specification with a risk level $\tau$.

\begin{problem}
\label{pb:set_inv}
Under Assumption~\ref{as:sample_belief}, given the robot motion described by ~(\ref{eq:sde_ro}) and the object dynamics by ~(\ref{eq:sde_obj}), a reference control input $ \bm{u}_\mathrm{ref} $, and a safe set $ \mathcal{C} $ defined over $ \bm{x} $ and the true object state $ \bm{o} $, the objective is to synthesize control inputs that remain close to $ \bm{u}_\mathrm{ref} $ while ensuring ${\Pr[H(\bm{x}_t,O_t) \geq 0] \geq 1-\tau}$, at every $t>t_0$ with level $\tau$. 
\end{problem}

The chance constraint $\Pr[H \geq 0] \geq 1-\tau$, however, is intractable in general non-Gaussian settings. To address this, we seek to ensure the tail risk measurements $\rho_\tau \geq0$, thereby enforcing the chance constraint, as shown in~(\ref{eq:chance_constr}). However, computing $\rho_\tau(H)$ directly is also difficult as we only know the samples of $H$. In the next section, we will address this challenge using the sample-based lower bounds on tail risk measures.



\section{Risk-aware Control}
\label{sec:method}
We present our approach to address Problem~\ref{pb:set_inv} by first formulating the dynamics of the sample-based belief for the object. Next, we introduce a novel form of BCBF and define the safe set as its zero super-level set. We then show that this BCBF remains robust to distributional shifts within a specified bound. Finally, we develop a controller that guarantees forward invariance of the safe set while preserving distributional robustness.

\subsection{Belief Dynamics of the Moving Object}
\label{sec:belief_dyn}
Given the belief samples of the moving object, we can define its belief state as 
\begin{align}
\label{eq:belief_state}
    \bm{b} = \begin{bmatrix} \bm{o}^{(1)} & \dots & \bm{o}^{(N)} \end{bmatrix}^\top \in \mathcal{B} \subseteq \mathbb{R}^{N \cdot n_o}.
\end{align}
As all the samples follow the same stochastic dynamics defined in~(\ref{eq:sde_ro}), the continuous time stochastic dynamics of the belief state $\bm{b}$ can be obtained as
\begin{equation}
\label{eq:belief_dynamics}
\begin{aligned}
    \mathrm{d} \bm{b} &= 
    \begin{bmatrix}
        \bm{\xi}\left(\bm{o}^{(1)}\right) \\
        \vdots \\
        \bm{\xi}\left(\bm{o}^{(N)}\right)
    \end{bmatrix} \mathrm{d} t
    +
    \begin{bmatrix}
        \bm{d} \left(\bm{o}^{(1)}\right) \, \mathrm{d} \bm{w}^{(1)} \\
        \vdots \\
        \bm{d} \left(\bm{o}^{(N)}\right)\, \mathrm{d} \bm{w}^{(N)}
    \end{bmatrix} \\
    &:= \bm{\Xi}(\bm{b})\mathrm{d}t+\bm{D}(\bm{b}) \mathrm{d} \bm{W},
\end{aligned}
\end{equation}
where $\bm{D}(\bm{b})=\mathrm{BD}\left( \{ \bm{d} \}_{i=1}^N \right)$ is a constant block diagonal matrix and $\bm{W}$ denotes a $N \cdot n_w$ dimensional Brownian motion. We note that the dimension of the belief state can become very high as $N$ increases, making methods such as model predictive control (MPC) computationally expensive when applied directly to sample-based beliefs. Scenario-based MPC \cite{de2021scenario,de2023scenario} demonstrates computational efficiency, but can become conservative without a careful choice of the number of samples $N$ \cite{mustafa2023probabilistic}. By contrast, the proposed BCBF method remains both less conservative and computationally efficient, even when using a large number of samples, as shown in the experiment later.

\subsection{Belief CBF from Concentration Bounds}
\label{sec:safe_set}

We aim to ensure that $\rho_\tau(H) \geq 0$, given only a finite set of samples $h^{(i)}(\bm{x}, \bm{o}^{(i)})$ for $i = 1, \dots, N$. To achieve this, we leverage the lower bounds on tail risk measures as presented in Lemma~\ref{lm:lbs}. Specifically, we define a scalar function of the sample-based belief state $\bm{b}$ from Eq.~(\ref{eq:belief_state}):
\begin{align*}
\tilde{h}(\bm{x},\bm{b}) &= \underline{\rho_{\tau}}\left( h^{(1)}(\bm{x},\bm{o}^{(1)}),\,\dots,\,h^{(N)}(\bm{x},\bm{o}^{(N)})\right).
\end{align*}
We select $\tilde{h}(\bm{x},\bm{b})$ as the BCBF, leading to the definition of a new safe set:
\begin{equation}
    \begin{aligned}
        \tilde{\mathcal{C}} &:= \bigl\{\, (\bm{x},\bm{b}) \in \mathcal{X} \times \mathcal{B} \,\big|\; \tilde{h}(\bm{x},\bm{b}) \ge 0 \bigr\}, \\
        \partial \tilde{\mathcal{C}} &:= \bigl\{\, (\bm{x},\bm{b}) \in \mathcal{X} \times \mathcal{B} \,\big|\; \tilde{h}(\bm{x},\bm{b}) = 0 \bigr\}.
    \end{aligned}
    \label{eq:belief_safe_set}
\end{equation}
If this set is ensured to be forward invariant, we can indicate ${\rho_\tau(H) \geq0}$ is satisfied with a probability of at least $1-\delta$.

\subsection{Robustness against Distributional Shift}
\label{sec:robustness}
The provided safe set in~(\ref{eq:belief_safe_set}) accounts for the Monte Carlo error introduced by sampling. However, when employing Bayesian filters,  we can encounter mismatches in the transition or measurement models. These mismatches may cause a distributional shift between the true belief distribution $O_{\mathrm{true}}$ and the estimated belief distribution $O_{\mathrm{esti}}$, which in turn leads to a shift between $H_{\mathrm{true}}$ and $H_{\mathrm{esti}}$.

We assume $H_{\mathrm{true}}$ and $H_{\mathrm{esti}}$ differ by at most ${\ell \in [0, \tau)}$ in one-sided  Kolmogorov-Smirnov (KS) distance, i.e., ${\sup_{h} \mathrm{CDF}_{H_{\mathrm{true}}}(h) - \mathrm{CDF}_{H_{\mathrm{esti}}}(h) \leq \ell}$. This captures the undesired cases where $\mathrm{CDF}_{H_{\mathrm{true}}}$ is above $\mathrm{CDF}_{H_{\mathrm{esti}}}$, increasing the probability of $H_{\mathrm{true}} < 0$, similar to the example in Fig~\ref{fig:fig2b}. 
Applying Lemma~\ref{lm:dist_robust}, we define the $\ell$-robust safe set:
\begin{equation*}
    \begin{aligned}
        \tilde{\mathcal{C}}^\ell &:= \bigl\{\, (\bm{x},\bm{b}) \in \mathcal{X} \times \mathcal{B} \,\big|\; \tilde{h}^\ell(\bm{x},\bm{b}) \ge 0 \bigr\}, \\
        \partial \tilde{\mathcal{C}}^\ell &:= \bigl\{\, (\bm{x},\bm{b}) \in \mathcal{X} \times \mathcal{B} \,\big|\; \tilde{h}^\ell(\bm{x},\bm{b}) = 0 \bigr\}.
    \end{aligned}
\end{equation*}
with the $\ell$-robust BCBF as
\begin{align}
\label{eq:belief_cbf}
\tilde{h}^\ell(\bm{x},\bm{b}) &= \underline{\rho^{\ell}_{\tau}}\left( h^{(1)}(\bm{x},\bm{o}^{(1)}_{\mathrm{esti}}),\,\dots,\,h^{(N)}(\bm{x},\bm{o}^{(N)}_{\mathrm{esti}})\right),
\end{align}
where $\bm{b} = \begin{bmatrix} \bm{o}^{(1)}_{\mathrm{esti}} & \dots & \bm{o}^{(N)}_{\mathrm{esti}} \end{bmatrix}^\top$. In the previous sections, we simply use $\bm{o}^{(i)}$ for brevity, but we note that all the samples are drawn from the estimated belief in practice.




\subsection{Controller Synthesis}
\label{sec:control}
To synthesize control inputs to maintain $\tilde{\mathcal{C}}^{\ell}$ forward invariant while keeping close to reference $\bm{u}_\mathrm{ref}$, we formulate the following quadratic program:
\begin{equation}
\label{eq:cbf-qp}
\begin{aligned}
    \bm{u}^* &= \arg\min_{\bm{u} \in \mathcal{U}} (\bm{u} - \bm{u}_\mathrm{ref})^T \bm{Q} (\bm{u} - \bm{u}_\mathrm{ref}) \\
    \text{s.t.} \quad
    & \frac{\partial \tilde{h}^\ell}{\partial \bm{x}} \left(f(\bm{x}) + g(\bm{x})\bm{u}\right) 
    + \frac{1}{2} \operatorname{tr} 
    \left( \bm{\sigma}^\top \frac{\partial^2 \tilde{h}^\ell}{\partial \bm{x}^2} \bm{\sigma} \right)
    - \frac{\left( \frac{\partial \tilde{h}^\ell}{\partial \bm{x}} \bm{\sigma} \right)^2}{\tilde{h}^\ell} \\
    &+\frac{\partial \tilde{h}^\ell}{\partial \bm{b}} \bm{\xi}(\bm{b})
    + \frac{1}{2} \operatorname{tr} 
    \left( \bm{d}^\top \frac{\partial^2 \tilde{h}^\ell}{\partial \bm{b}^2} \bm{d} \right)
    - \frac{\left( \frac{\partial \tilde{h}^\ell}{\partial \bm{b}} \bm{d} \right)^2}{\tilde{h}^\ell} \\
    &\geq - \gamma \tilde{h}^\ell(\bm{x},\bm{b})^3,
\end{aligned}
\end{equation}
where $\gamma \geq 0$, and $\bm{Q}$ stands for the weighting matrix. We note that $\tilde{h}^\ell$ may be non-smooth due to the permutations introduced by the order statistics, which can lead to discontinuities in its gradient. However, as shown in Lemma~\ref{lm:non_smooth}, this non-smoothness does not affect the forward invariance property. Therefore, the control input $\bm{u^*}$ is sufficient to ensure forward invariance, as stated in Theorem~\ref{th:main}.
\begin{theorem}
\label{th:main}
Suppose that the controller obtained by solving~(\ref{eq:cbf-qp}) remains feasible at all times, and the non-smooth stochastic CBF $\tilde{h}^\ell(\bm{x},\bm{b})$ is defined as in Eq.~(\ref{eq:belief_cbf}). Then the safe set $\tilde{\mathcal{C}}^\ell$ is forward invariant.
\end{theorem}
\begin{proof}
Following Theorem~\ref{th:nszcbf}, we substitute the object state $\bm{o}$ with the belief state $\bm{b}$ as defined in Eq.~(\ref{eq:belief_state}), the system~(\ref{eq:sde_obj}) to the belief dynamics~(\ref{eq:belief_dynamics}), and the function $h$ to $\tilde{h}^{\ell}$. Provided the controller remains feasible, it follows that $\tilde{\mathcal{C}}^\ell$ remains invariant with probability 1.
\end{proof}



\section{Example Applications and Results}
We validate our approach through numerical simulations in two example applications.\footnote{Code available at: \url{https://github.com/KTH-RPL-Planiacs/sample_based_bcbf}} These applications focus on underwater scenarios where both stochastic dynamics and uncertain states naturally occur. Specifically, we consider a planar AUV, modeled as a stochastic unicycle, as in \cite{joris10886238}:
\begin{equation*}
\begin{bmatrix} 
    \mathrm{d} p_x \\ 
    \mathrm{d} p_y \\ 
    \mathrm{d} \theta 
\end{bmatrix} = 
\begin{bmatrix} 
    \mathrm{cos}(\theta) & 0 \\ 
    \mathrm{cos}(\theta) & 0 \\ 
    0 & \mathrm{sin}(\theta) 
\end{bmatrix} 
\bm{u}\mathrm{d}t + \bm{\sigma} \, \mathrm{d} \bm{z},
\end{equation*}
where the state vector is defined as $\bm{x} = [p_x, p_y, \theta]^\top$, consisting of the 2D position $\bm{p} = [p_x, p_y]^\top$ and orientation $\theta$, and the control input is given by $\bm{u} = [u_v, u_\omega]^\top$, representing the linear and angular speeds. The diffusion term $\bm{\sigma}=\mathrm{diag}([0.03, 0.03,0.01])$. The robot is modeled with a circular footprint of radius $r_{\mathrm{e}}$, while the distance between the rear axle and the center is denoted as $s_\mathrm{e}$. Meanwhile, we model the moving object as a stochastic single integrator 
\begin{equation}
\label{eq:single_inte}
\begin{bmatrix} 
    \mathrm{d} q_x \\ 
    \mathrm{d} q_y 
\end{bmatrix} = 
\begin{bmatrix} 
    v_x \\ 
    v_y \\ 
\end{bmatrix} 
\mathrm{d}t + \bm{d} \, \mathrm{d} \bm{w},
\end{equation}
where the state is $\bm{o} = [q_x, q_y]^\top$, meaning the 2D position, and the diffusion term is $\bm{d} = \mathrm{diag}([0.1, 0.1])$. We have the knowledge of an estimated velocity $[v_x,v_y]^\top$ from the perception module. The object also has a circular footprint of radius $r_{\mathrm{o}}$.

With a prior sample-based belief of the object, we assume the samples are propagated continuously according to~(\ref{eq:single_inte}) and no measurement updates occur, i.e. no sudden discrete update in the belief. This assumption is motivated by real-world challenges such as temporary low visibility in underwater scenarios. Despite the absence of measurement updates, a well-designed controller should be capable of satisfying the specifications until sufficient information is available to reduce uncertainty. Our approach also accommodates scenarios where measurement updates become available, under an additional assumption as \cite{vahs2024risk}[Assumption 1].




\begin{figure}[t]\captionsetup[subfigure]{font=scriptsize}
\centering
\subfloat[]{\includegraphics[width=0.49\columnwidth]{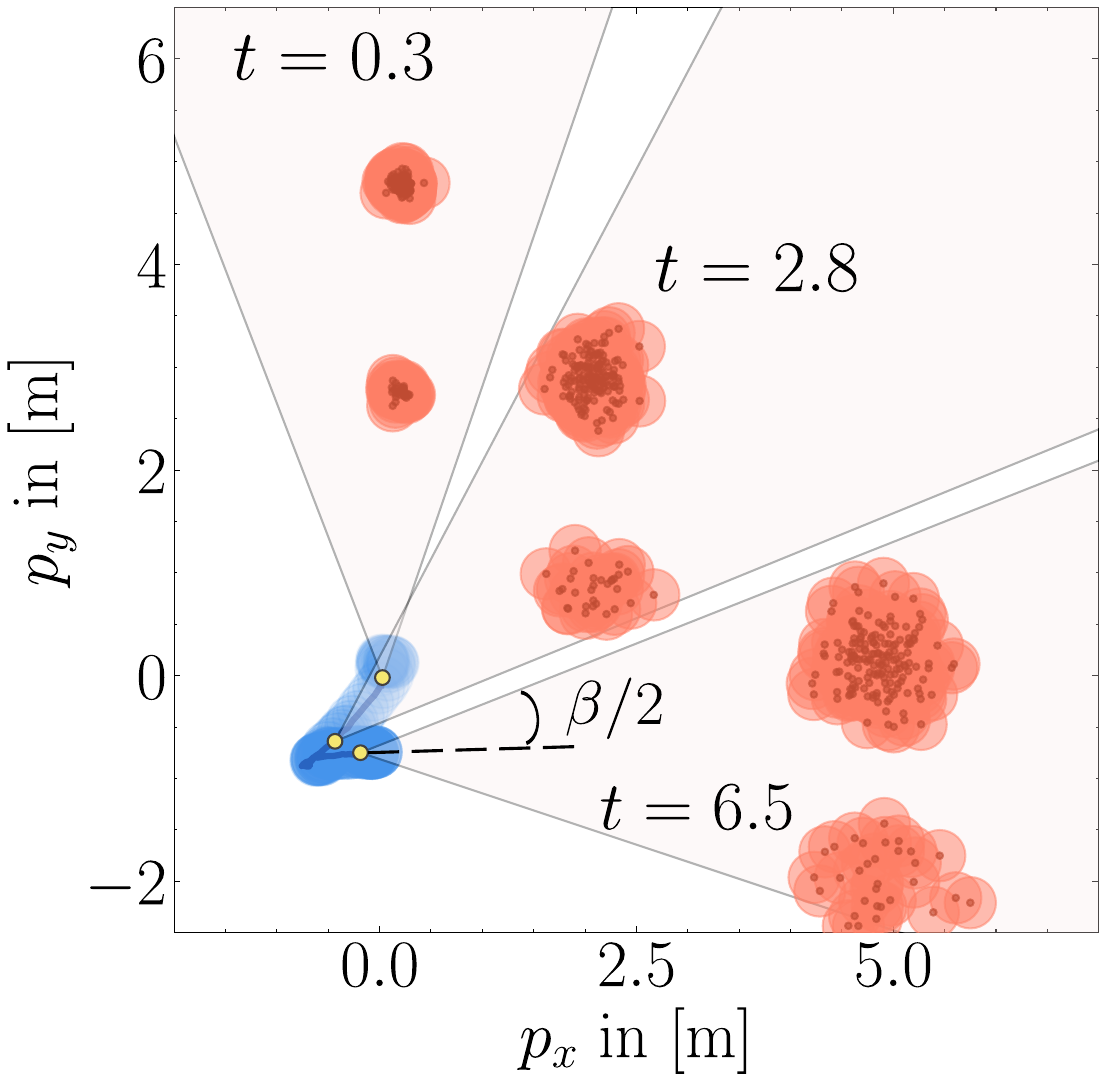} \label{fig:exp_tr_sce}}
\subfloat[]{\includegraphics[width=0.45\columnwidth]{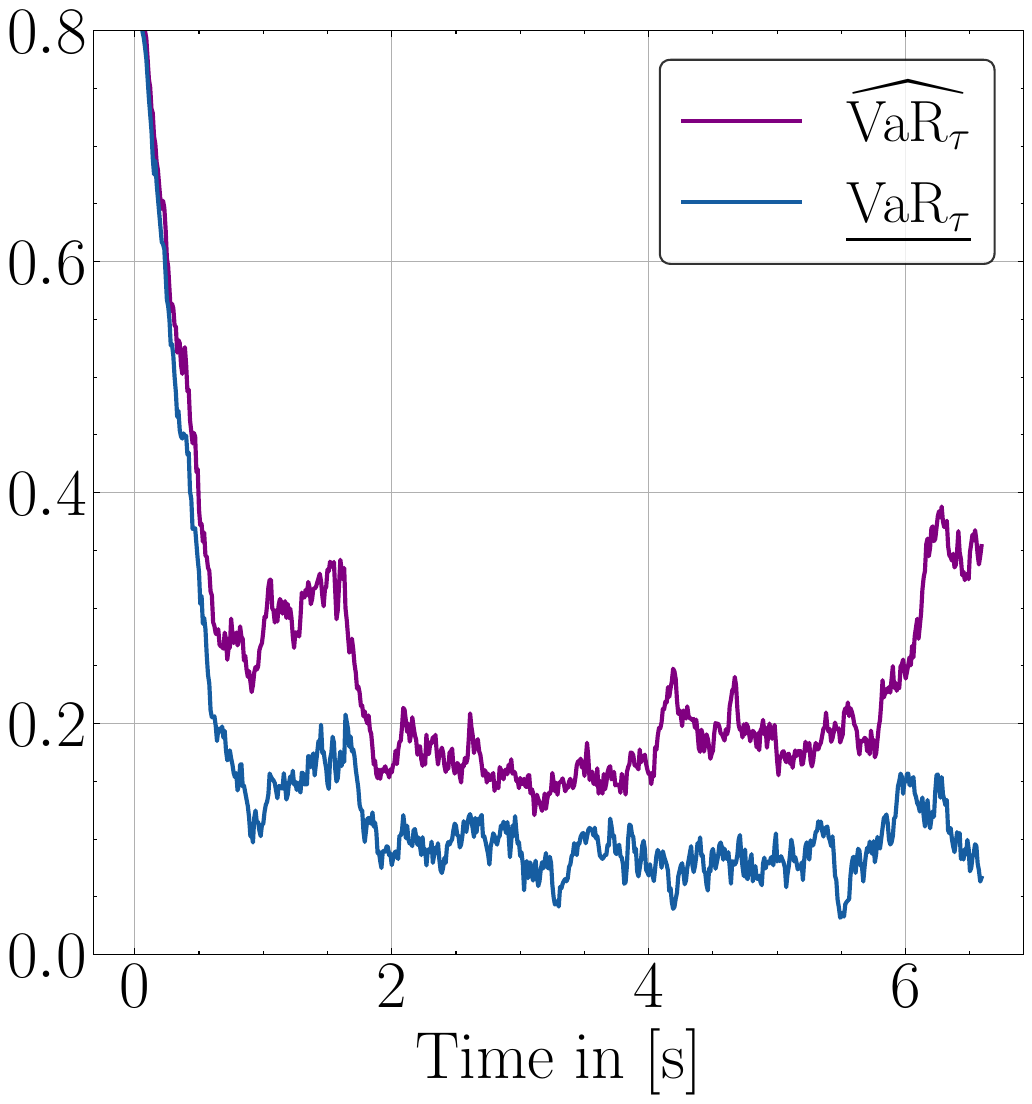} \label{fig:exp_tr_h}}
\caption{
(a): Snapshot of the object tracking scenario. The trajectory of the robot is shown in blue. The FoV (pink sectors), the robot (yellow dots), and the belief samples of the object (red dots) are taken at $t=0.3 \mathrm{s}$, $t=2.8\mathrm{s}$, and $t=6.5\mathrm{s}$. The blue and red shadows represent the circular footprints of the robot and the object, respectively. (b): Evolution of $\widehat{\mathrm{VaR_\tau}}$ and $\underline{\mathrm{VaR_\tau}}$ over the course of the simulation.
}
\vspace{-0.5cm}
\label{fig:exp_tr}
\end{figure}

\subsection{Object Tracking}
\label{sec:exp_tr}
We first apply the proposed BCBF in an object tracking scenario, where the objective is to ensure that the moving object remains within the robot's FoV. Meanwhile, a performance controller attempts to stabilize the robot at its initial position close to $[0,0]^\top$. We model the FoV as a limited angular sector with an amplitude $\beta=\SI{40}{\degree}$ centered at local $x$-axis, as illustrated in Fig.~\ref{fig:exp_tr_sce}. Following the approach in \cite{catellani2023distributed}, we compute the coordinates of the object in the robot's local frame, denoted as $[{}^p q_x, {}^p q_y]^\top$ and define the safe set as the intersection of two zero super-level sets of the function:
\begin{equation*}
\begin{aligned}
    h_i(\bm{x},\bm{o}) =& \tan(\beta/2)\cdot{}^p q_x-r_{\mathrm{o}}/\mathrm{cos}(\beta/2) \\
    &+ (-1)^{i}\cdot {}^p q_y , \quad i \in \{1,2\}.
\end{aligned}
\end{equation*}
We draw $N = 200$ samples drawn from an initial belief of the object, given as a skewed multimodal Gaussian $p(\bm{o}_{t_0}) = 0.85 \cdot \mathcal{N}(\bm{o}_1, \mathbf{\Sigma}_1) + 0.15 \cdot \mathcal{N}(\bm{o}_2, \mathbf{\Sigma}_2)$. The object has velocity $[v_x, v_y]^\top = [0.75, -0.75]^\top$. We use $\underline{\mathrm{VaR}_{0.1}}$ in the BCBF $\tilde{h}_i$, as we look at the worst \SI{10}{\percent} quantile.

As shown in Fig.~\ref{fig:exp_tr_sce}, the object's samples move from left top to right bottom, while becoming more dispersed over time due to stochastic motion. Consequently, the AUV moves backward to keep the sample-based distribution within its FoV, even though the performance controller guides the AUV toward its initial position and the belief distribution is highly skewed. In Fig.~\ref{fig:exp_tr_h}, we observe that for ${i\in\{1,2\}}$, the lower bound $\mathrm{min}_i\underline{\mathrm{VaR}_{\tau}}(h^{(1)}_i,\dots,h^{(N)}_i)$ stays positive, and is consistently smaller than the empirical value $\mathrm{min}_i{\widehat{\mathrm{VaR}_{\tau}}(h^{(1)}_i,\dots,h^{(N)}_i)}$ computed from the empirical $\mathrm{CDF}$ constructed by samples. This indicates $\Pr[H_i \geq 0] \geq 1-\tau$ is satisfied in our simulation.

\subsection{Dynamic Collision Avoidance}
\label{sec:exp_cv}
We then apply BCBF to a dynamic collision avoidance scenario. The objective is to avoid colliding with the object, while the performance controller tries to drive the robot to a target position $[3.0, 3.0]^\top$, as shown as a green star in Fig.~\ref{fig:exp_colav_comp} and~\ref{fig:exp_colav_dist}.  We obtain $N$ samples from the initial belief of the object, which is a skewed Gaussian mixture $p(\bm{o}_{t_0}) = 0.7 \cdot \mathcal{N}(\bm{o}_1, \mathbf{\Sigma}_1) + 0.15 \cdot \mathcal{N}(\bm{o}_2, \mathbf{\Sigma}_2) + 0.15 \cdot \mathcal{N}(\bm{o}_3, \mathbf{\Sigma}_3)$. The object has velocity $[v_x, v_y]^\top = [-0.75, -0.75]^\top$. Here, we use the similar safe set as in \cite{huang2023obstacle}, defined as ${h(\bm{x},\bm{o}) = \|\hat{\bm{p}} \|_2 - (r_\mathrm{e}+r_\mathrm{o})}$, where ${\hat{\bm{p}}=[p_x-q_x+s_\mathrm{e} \cdot \mathrm{cos} (\theta), p_y-q_y+s_\mathrm{e} \cdot \mathrm{sin} (\theta)]^\top}$.

\begin{figure}[t]\captionsetup[subfigure]{font=scriptsize}
\centering
\subfloat[]{\includegraphics[width=0.48\columnwidth]{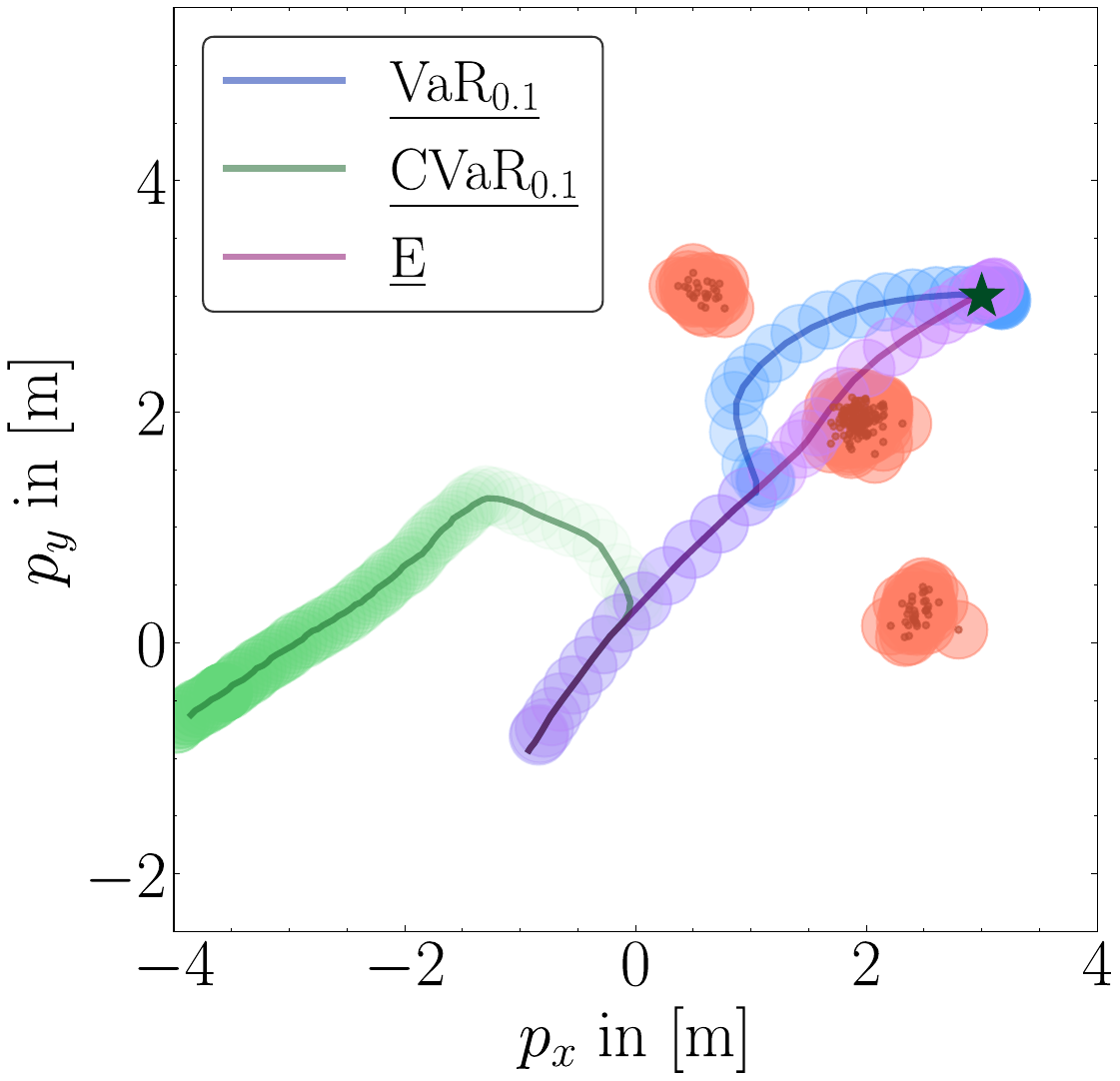} \label{fig:exp_colav_comp}}
\subfloat[]{\includegraphics[width=0.48\columnwidth]{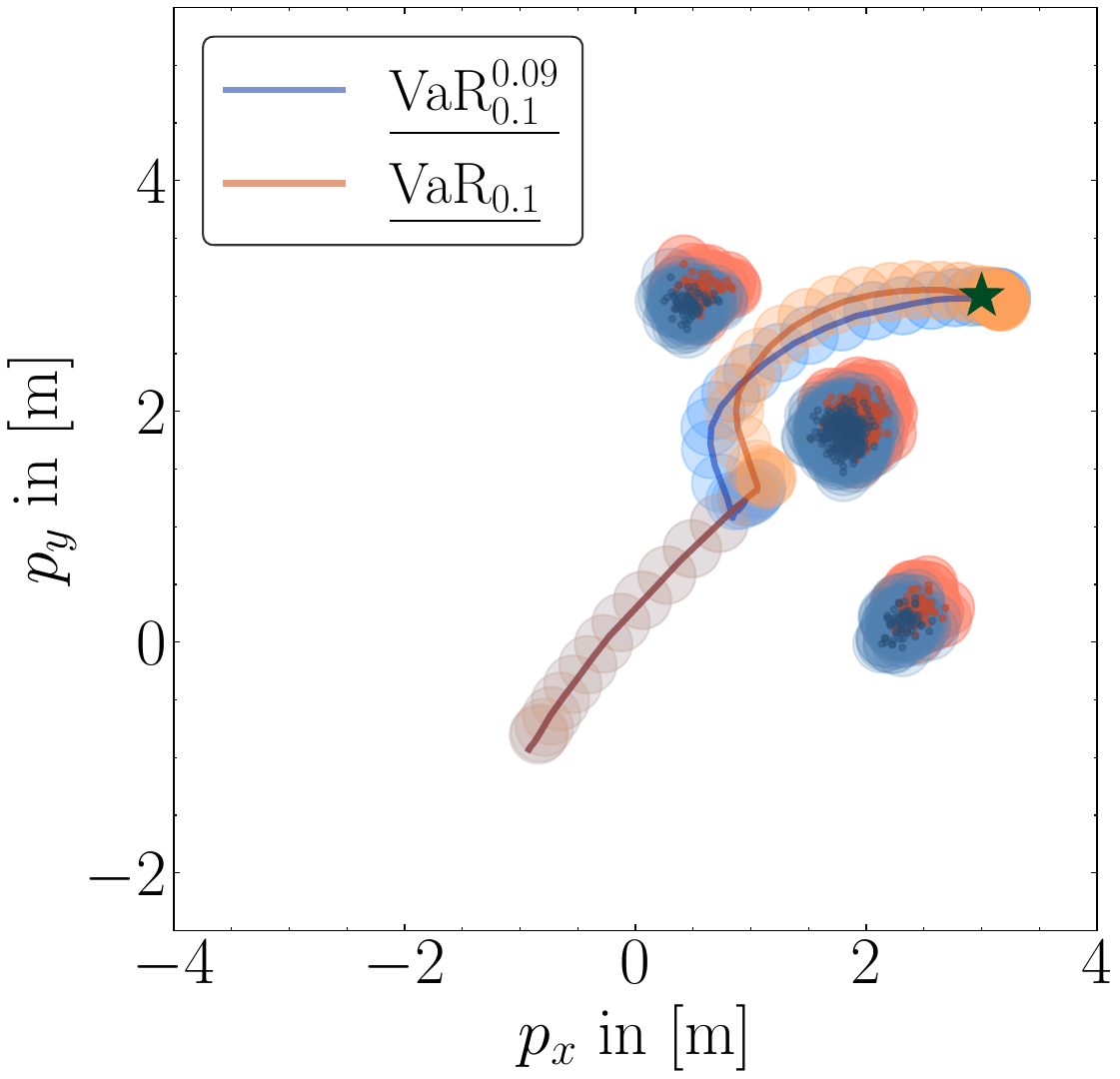} \label{fig:exp_colav_dist}}
\caption{(a): Snapshot of the collision avoidance scenario. The belief samples (red dots) are taken at time $t =\SI{0.76}{\second}$. (b): Snapshot illustrating a distributional shift. The true samples (dark blue dots) differ from those used by the controller (red dots), while both are also taken at time $t =\SI{0.76}{\second} $.}
\label{fig:exp_colav}
\vspace{-0.5cm}
\end{figure}

\subsubsection{Performance Comparison of Risk Measures}
\label{sec:exp_cv_bm}
We consider using both $\underline{\mathrm{VaR}_{\tau}}$, and $\underline{\mathrm{CVaR}_{\tau}}$, evaluated at $\tau = 0.1$ and $\delta = 0.05$. Additionally, we benchmark performance against the expected value $\underline{\mathbb{E}}$. We run 100 independent simulations for three different sample sizes: $N = 200$, $1000$, and $5000$. For each simulation run, we randomize both the initial state of the robot and the mean values of the Gaussian mixture distribution $p(\bm{o}_{t_0})$. The results of these experiments are summarized in Table~\ref{tb:cv_benchmark}. The number of successes, collisions, and timeouts is reported. $T_\mathrm{avg}$ represents the average computation time of~(\ref{eq:cbf-qp}) in milliseconds, including computing the parameters and solving the QP. Furthermore, we highlight a representative case and present its results in Fig.~\ref{fig:exp_colav_comp}. We can observe that using $\underline{\mathbb{E}}$ results in unsafe behavior as it fails to account for the tail events. While $\underline{\mathrm{CVaR}_{0.1}}$ ensures safety, it can be overly conservative. In contrast, $\underline{\mathrm{VaR}_{0.1}}$ offers a balance between safety and task achievement.


Table~\ref{tb:cv_benchmark} shows that $\mathbb{E}$ leads to a higher number of collisions. In contrast, the tail risk measures ensure safety, as they account for the skewed and multimodal nature, rather than relying solely on the mean. Furthermore, $\underline{\mathrm{VaR}_{0.1}}$ achieves more successes with no timeouts, indicating that it is less conservative than $\underline{\mathrm{CVaR}_{0.1}}$ while still ensuring safety.  We also notice that the controller employing $\underline{\mathrm{CVaR}_{\tau}}$ can be made less conservative by selecting a higher risk level $\tau$. However, this may compromise safety guarantees and reduce interpretability. Additionally, $\underline{\mathrm{VaR}_{\tau}}$ offers advantages due to its simpler mathematical formulation in Eq.~(\ref{eq:lower_var}), making it easier to implement and requiring less time to compute. Our implementation of BCBF with $\underline{\mathrm{VaR}_{\tau}}$ achieves a computation frequency of approximately \SI{1}{\kilo\hertz}, even for $5000$ samples. We compute the parameters of the QPs using JAX \cite{bradbury2018jax}.
The QPs are solved using ProxQP \cite{bambade2023proxqp} and its official implementations. All computations are performed on a laptop with an Intel Core i7-13700H CPU, 32 GB of RAM, and an NVIDIA RTX 4070 GPU with 8 GB of memory.


\subsubsection{Distributional Robustness}
\label{sec:exp_cv_dr}
While the object's estimated velocity is $[v_x, v_y]^\top = [-0.75, -0.75]^\top$, we consider a case where the true velocity stays to be the same direction, but is \SI{20}{\percent} faster. This model mismatch induces a distributional shift in the sample-based belief, as illustrated in Fig.~\ref{fig:exp_colav_dist}. To mitigate this issue, we employ the $\ell$-robust BCBF $\underline{\mathrm{VaR}_{0.1}^{\ell}}$. 
We set $\ell = 0.09$, as ${\ell\in[0,\tau)}$.
We choose $N=500$ and summarize the results in Table~\ref{tb:cv_dr}. We can observe that the robust BCBF effectively enhances safety, while collisions happen more frequently with the original BCBF.

\begin{table}[t]
\centering
\begin{threeparttable}  
    \caption{Benchmark Results with 100 Simulations}
    \label{tb:cv_benchmark}

\begin{tabular}{@{}clcccc@{}}
\toprule
$N$ & Method & Success & Collision & Timeout\tnote{*} & $T_{\text{avg}}$ (ms) \\ \midrule
& BCBF-\(\underline{\mathrm{VaR}_{0.1}}\) & \textbf{97/100} & 3/100 & \textbf{0/100} & \textbf{0.668} \\
200 & BCBF-\(\underline{\mathrm{CVaR}_{0.1}}\) & 91/100 & \textbf{0/100} & 9/100 & 0.708 \\
& BCBF-\(\underline{\mathbb{E}}\) & 46/100 & 54/100 & 0/100 & 0.710 \\ \midrule
& BCBF-\(\underline{\mathrm{VaR}_{0.1}}\) & \textbf{98/100} & 2/100 & \textbf{0/100} & \textbf{0.769} \\
1000 & BCBF-\(\underline{\mathrm{CVaR}_{0.1}}\) & 89/100 & \textbf{1/100} & 10/100 & 0.869 \\
\textbf{} & BCBF-\(\underline{\mathbb{E}}\) & 30/100 & 70/100 & 0/100 & 0.802 \\ \midrule
& BCBF-\(\underline{\mathrm{VaR}_{0.1}}\) & \textbf{99/100} & \textbf{1/100} & \textbf{0/100} & \textbf{1.146} \\
5000 & BCBF-\(\underline{\mathrm{CVaR}_{0.1}}\) & 94/100 & \textbf{1/100} & 5/100 & 1.554 \\
& BCBF-\(\underline{\mathbb{E}}\) & 28/100 & 72/100 & 0/100 & 1.586 \\ \bottomrule
\end{tabular}

    \begin{tablenotes}
        \item[*] A timeout occurs when the robot neither collides nor successfully reaches the target within a reasonable period, indicating that it might be stuck. We set the maximum simulation time to \SI{10}{s} in our experiment.
    \end{tablenotes}
\end{threeparttable}
\end{table}

\begin{table}[ht]
\begin{center}
    \caption{Performance Under Distributional Shift}
    \label{tb:cv_dr}
\begin{tabular}{@{}lccc@{}}
\toprule
Method                                       & \multicolumn{1}{l}{Success} & \multicolumn{1}{l}{Collision} & \multicolumn{1}{l}{Timeout} \\ \midrule
BCBF-$\underline{\mathrm{VaR}_{0.1}}$        & 89/100                      & 11/100                        & 0/100                       \\ \midrule
BCBF-$\underline{\mathrm{VaR}_{0.1}^{0.09}}$ & \textbf{99/100}                       & \textbf{1/100}                          & 0/100                       \\ \bottomrule
\end{tabular}
\end{center}
\vspace{-0.7cm}
\end{table}


\section{Conclusion}
This paper presented a BCBF framework for ensuring the safety of robots operating in dynamic and uncertain environments. By incorporating provable concentration bounds on tail risk measures, the proposed method works effectively under sample-based beliefs of the environment and stochastic dynamics. Experimental validation in simulated underwater scenarios demonstrates real-time feasibility and robust collision avoidance under distribution shifts. 
We simplify the scenario by modeling the robot as a planar unicycle, which serves as an illustrative example. In future work, we aim to extend the approach to higher-order CBFs and underwater robots operating in 3D.


\section*{Appendix}
\subsection{Proof of the Lemma~\ref{lm:non_smooth}}
\label{ap:proof_nszcbf}
\begin{proof}
Note $h(\bm{x},\bm{o})$ evolves under $    {\mathrm{d}h \;=\; \bar{\mu}\,\mathrm{d}t \;+\; \bm{\bar{\sigma}}\,\mathrm{d}\bm{\bar{w}}}$,
where $\mathrm{d}\bm{\bar{w}}=\begin{bmatrix}
    \mathrm{d}\bm{z} & \mathrm{d}\bm{w}
\end{bmatrix}^\top$. By It\^{o}'s lemma \cite[Lemma 1]{clark2021control}, we have
\begin{equation*}
    \begin{aligned}
        \bar{\mu} =& \frac{\partial h}{\partial \bm{x}} \left(\bm{f}(\bm{x}) + \bm{g}(\bm{x})\bm{u}\right) 
        + \frac{1}{2} \operatorname{tr} 
        \left( \bm{\sigma}^\top \frac{\partial^2 h}{\partial \bm{x}^2} \bm{\sigma} \right) \\
        & + \frac{\partial h}{\partial \bm{o}} \bm{\xi}(\bm{o})
        + \frac{1}{2} \operatorname{tr} 
        \left( \bm{d}^\top \frac{\partial^2 h}{\partial \bm{o}^2} \bm{d} \right), \\
        \bm{\bar{\sigma}} =& 
        \begin{bmatrix} \frac{\partial h}{\partial \bm{x}}\,\bm{\sigma}(\bm{x})  & \frac{\partial h}{\partial \bm{o}}\,\bm{d}(\bm{o}) \end{bmatrix}.
    \end{aligned}
\end{equation*}
We define the auxiliary function $B=1/h$. Then by It\^{o}'s lemma, the drifting term $\hat{\mu}$ of $B$ is
\begin{equation*}
    \hat{\mu} = \frac{\partial B}{\partial h} \bar{\mu} 
    + \frac{1}{2} \lVert \bm{\bar{\sigma}} \rVert_2^2 \frac{\partial^2 B}{\partial h^2}
    = -h^{-2} \bar{\mu} + h^{-3} \lVert \bm{\bar{\sigma}} \rVert_2^2.
\end{equation*}
Multiplying~(\ref{ineq:nszcbf}) by $-1/h^2$ and substitute in $\bar{\mu}$ and $\bm{\bar{\sigma}}$ yield $\hat{\mu} \leq \alpha (h (\bm{x}, \bm{o}))$. Plugging in the expression for $\hat{\mu}$ in terms of $\bm{x}$ and $\bm{o}$, we have
\begin{equation*}
\begin{aligned}
    & \frac{\partial B}{\partial \bm{x}} (f(\bm{x}) + g(\bm{x}) \bm{u}) 
    + \frac{1}{2} \operatorname{tr} \left[ \bm{\sigma}^\top \frac{\partial^2 B}{\partial \bm{x}^2} \bm{\sigma} \right] \\
    & + \frac{\partial B}{\partial \bm{o}} \bm{\xi}(\bm{o})
        + \frac{1}{2} \operatorname{tr} 
        \left( \bm{d}^\top \frac{\partial^2 B}{\partial \bm{o}^2} \bm{d} \right) 
     \leq \alpha \left(h (\bm{x})\right).
\end{aligned}
\end{equation*}
Hence, $B$ is a non-smooth stochastic CBF  \cite[Definition 3]{vahs2024non} with $\alpha_1(h)=\alpha_2(h)=h$, ensuring the forward invariance.
\end{proof}




\section*{Acknowledgement}
We thank Ignacio Torroba, Joris Verhagen, Diego Martinez, Yifei Dong, Ming Li, Yixi Cai, Hanqing Xiang, Luyao Zhang, and the Planiacs for the helpful discussions.


\bibliographystyle{IEEEtran}
\bibliography{refs.bib}

\end{document}